\documentclass[a4paper,11pt]{llncs}

\usepackage{amssymb}
\def\it{\textit}
\def\bf{\textbf}
\def\tt{\texttt}
\def\ul{\underline}
\def\mb{\mathbf}

\def\Q{\mathrm{Q}}
\def\bb{\mathbb}
\def\mc{\mathcal}

\def\ds{\displaystyle}

\everymath{\displaystyle}

\begin{document}

\mainmatter

\title{Expressibility at the machine level versus structure level:
ESO universal Horn Logic and the class \bf{P}}

\titlerunning{Machine level and structure level separation}

\author{Prabhu Manyem}

\institute{Department of Mathematics,
Shanghai University,
Shanghai 200444, China}

\maketitle

\begin{abstract}
We show that ESO universal Horn logic (existential second logic where the
first order part is a universal Horn formula) is insufficient to capture
\bf{P}, the class of problems decidable in polynomial time. This statement is
true in the presence of a successor relation in the input
vocabulary. We provide two proofs --- one based on reduced products of two
structures, and another based on approximability theory (the second proof
is under the assumption that \bf{P} is not the same as \bf{NP}). ~We show that
the difference between the results here and those in \cite{gradel91},
is due to the fact that the expressions this paper deals with are
at the ``structure level", whereas the expressions in \cite{gradel91}
are at the ``machine level" --- a case of \it{Easier done than
said}.
\end{abstract}

% This abstract is for Ginny Whatarau at Victoria University of
% Wellington, NZ, sent on May 29, 2012, for the ALC 2011 proceedings:
% 
% \begin{abstract}
% We show that ESO universal Horn logic (existential second logic where
% the first order part is a universal Horn formula) is insufficient to
% capture \bf{P}, the class of problems decidable in polynomial time.
% This statement is true in the presence of a successor relation in the
% input vocabulary. We provide two proofs --- one based on reduced
% products of two structures, and another based on approximability theory
% (the second proof is under the assumption that \bf{P} is not the same
% as \bf{NP}).  ~The difference between the results here and those in
% (Gr\"{a}del 1991) is due to the fact that the expressions this paper
% deals with are at the ``structure level", whereas the expressions in
% (Gr\"{a}del 1991) are at the ``machine level" --- a case of \it{Easier
% done than said}.
% \end{abstract}



\section{Introduction}

We work with decision problems derived from optimization problems.
The reader is assumed to have some background in Finite Model
Theory.
If not, the book by Ebbinghaus and Flum \cite{EF99} serves as a good
introduction.
Publications \cite{gradelAnd7others} and \cite{immerman} are also
relevant material for this line of research.
Our main result is the following theorem:
\begin{theorem}
Existential second order logic with a first order part that is
universal Horn (or simply, ESO universal Horn) is insufficient to capture
the class \bf{P}, the class of problems decidable in polynomial time,
even assuming that the input structures come with a successor relation.
$\hfill \Box$
\label{thm:esoHierarchy}
\end{theorem}
% \begin{theorem}
% ESO universal logic (where the first order part is a universal sentence)
% is insufficient to capture \bf{NP} regardless of
% whether the input structure is ordered.
% \label{thm:faginFalse}
% \end{theorem}
We provide two proofs of Theorem \ref{thm:esoHierarchy}: (a) one based on reduced
products, and (b) a second proof based on approximability theory (this
proof assumes that \bf{P} $\neq$ \bf{NP}).

We first introduce some notation and definitions.

\begin{definition}
\bf{ESO universal Horn} is defined as existential second
order (ESO) logic where the first order part is a universal Horn formula.
$\hfill \Box$
\end{definition}

\begin{problem}
\bf{Maximum Matching}.
Given an undirected graph $G = (V, E)$ and a non-negative integer $K$, is
there a matching in $G$ such that the number of matched edges is at least
$K$?
Assume that $G$ has (i) no self-loops; and (ii) at most one edge between
any pair of vertices.
$\hfill \Box$
\label{def:matching}
\end{problem}

\begin{remark}
(\bf{A note about K}).
~Readers may ask, ``How is $K$ represented in the input? In binary or
unary?".
This may be relevant for problems with large data, such as
Maximum Flow or Minimum Cost Flow, where edge weights are part of the input.
However, the issue is irrelevant for problems such as Maximum Matching where
the solution value lies in the $[0, t]$ range, where $t = \lfloor n/2
\rfloor$ and $n$ is the number of vertices.  If $K > t$, then the
solution is 
immediately disqualified --- there cannot be a solution to the
Maximum Matching 
problem with the number of matched edges higher than $t$.
If $K \le t$, the input size for $K$ is at most that of $t$, whether in
binary or unary.
$\hfill \Box$
\end{remark}

For simplicity, assume the following for Problem \ref{def:matching}:
\begin{assum}
~~$K \le n$, where $n$ is the number of vertices in the input graph of
the problem instance.  Also,
throughout this article, $K$ is a part of the input instance; it can vary
from instance to instance; it is \it{not} a fixed constant over all
instances.
$\hfill \Box$
\label{ass:KleqN}
\end{assum}

\subsection{BFC and OFC}\label{sec:bfcOfc}

We represent decision versions of optimization problems as a conjunction
of a single objective function constraint (OFC) and a set of basic
feasibility constraints (BFC).
~The motivation for this comes from a general mathematical programming
framework where optimization problems are expressed in the form:
\begin{equation}
\begin{array}{rl}
& \mbox{Maximize (or Minimize) $f(\mb{x})$  (the \it{objective
function})}, \\[2mm]
& \mbox{subject to the following constraints:} \\[2mm]

& \mbox{$\ds g_i(\mb{x}) \le 0$, ~ $1 \le i \le m$}, \\[1mm]
& \mbox{$\ds h_j(\mb{x}) = 0$, ~ $1 \le j \le p$}, \\[2mm]

& \mbox{and $\mb{x} \in X$ (for example, $X = \bb{R}^n$)}.
\end{array}
\end{equation}
Additionally, for the decision problem, we are given a number 
$K$ as a part of the input.

Above, the BFC comprises the $m+p$ constraints
$g_i(\mb{x}) \le 0$ and $h_j(\mb{x}) = 0$, and the OFC comprises the
single constraint $f(\mb{x}) \ge K$ for maximization problems ($f(\mb{x})
\le K$ for minimization problems).

\section{Proof of Theorem \ref{thm:esoHierarchy} based on reduced products}

% The proof here is adapted from the one given by Dawar
% \cite{dawarPersonal}, which showed that the cardinality constraint $|F|
% \ge K$ cannot be expressed by a Horn formula, 

The proof is a counterexample. It relies on the fact that universal Horn
formulae are preserved under reduced products.
That is, if $\psi$ is a universal Horn formula and
structures $\mb{A}$ and $\mb{B}$ satisfy $\psi$, then so does their
reduced product $\mb{A.B}$.
~ From the definition of Problem  \ref{def:matching}, note that any feasible solution to
\it{Maximum Matching} consists of at least $K$ matched edges.  
 
\begin{lemma}
Maximum Matching cannot be expressed in ESO universal Horn logic,
even assuming that the input structures come with a successor relation.
$\hfill \Box$
\label{lem:matchingReducedProducts}
\end{lemma}
\begin{proof}
Suppose we have a signature divided into two parts:
\newline
A set of relation symbols $\Sigma = \{S_1, \cdots, S_m, F\}$, and a
separate set $\rho$ of relation and constant symbols.
$F(u,v)$ is true iff $(u,v)$ is a matched edge.

The relations in $\rho$ are first order (input), hence the Horn restriction
does not apply to these.  On the other hand, the relations in $\Sigma$
are second order (quantified), and hence the Horn condition applies to
these relations.  In particular, it applies to $F$, which is unknown; it
is not a part of the input.

Let us introduce this definition (henceforth referred to as Property
\ref{def:maxMatchProp}):
\begin{definition}
(Maximum Matching property).
Given $G = (V,E)$ and an integer $K$, is there a matching $F \subseteq E$ such that $|F| \ge K$?
$\hfill \Box$
\label{def:maxMatchProp}
\end{definition}

Let \bf{A} and \bf{B} be two input structures on the same universe
 $\{0, 1, \cdots, 2\theta\}$ ($\theta \in \mathbb{N}$) with a successor
relation; 
furthermore, for all symbols in the signature $\rho$, let \bf{A} and
\bf{B} agree on the interpretation of the symbols; however, they may
differ on relations in $\Sigma$.

Now define a product structure \bf{A.B}
to be the structure on the domain $\{0, 1, \cdots, 2\theta\}$ with a linear
ordering, in which every symbol in the signatures $\Sigma$ and $\rho$ is
interpreted as the \it{intersection} of the two interpretations in \bf{A}
and \bf{B}.
~Then it is easy to show that any Horn formula $\phi$ that is satisfied
in both \bf{A} and \bf{B} is also satisfied in \bf{A.B}.
~(This is exactly by the argument that shows that Horn formulae are closed
under direct products; see  Chapter 9, Page 417, of \cite{hodges93},
or see Page 493 of \cite{hodgesBookChapter}.)

Then suppose we let \bf{A} and \bf{B} be structures with, say,
$2\theta+1$ vertices $\{0, 1, \cdots, 2\theta\}$.
~Assume that in both structures, we have edges from the set 
$\{(i-1, i) ~|~ 1 \le i \le 2 \theta\}$.

In \bf{A}, we let $F$ consist of the $\theta$ edges
$\{ (0,1), (2,3), \cdots, (2\theta-4, 2\theta -3), (2\theta-2, 2\theta -1) \}$;
~and in \bf{B}, we let $F$ consists of the $\theta$ edges
$\{ (1,2), (3,4), \cdots, (2\theta-3, 2\theta -2), (2\theta-1, 2\theta) \}$.
~Thus $F$ indeed defines a matching in both structures.

Also, we let $K = \theta$ in both structures.
~So, Property \ref{def:maxMatchProp} is true in both \bf{A} and \bf{B}.
~That is, there exists a matching of size at least $\theta$ in \bf{A} and
\bf{B}.

However, Property \ref{def:maxMatchProp} is false in \bf{A.B}, since in this structure,
$F$ is empty, even though the lower bound $K$ is still equal to $\theta$.
~It follows that Property \ref{def:maxMatchProp} cannot be expressed as a Horn
formula, and hence we conclude that \it{Maximum Matching},
a problem in the class \bf{P}, cannot be expressed in ESO universal Horn
logic.
$\hfill \Box$
\end{proof}

Notes on the above proof:
\begin{itemize}
\item
Readers may ask, \it{Is it not necessary to look for a solution $F$ (a
matching) in the product structure?} --- the answer is no, it is \it{not}
necessary.
~If Property \ref{def:maxMatchProp} can indeed be expressed in ESO
universal Horn logic, and assuming that the property is satisfied in
\bf{A} and \bf{B}, then it must be satisfied in the product structure
\bf{A.B}, according to the results in \cite{hodges93} and
\cite{hodgesBookChapter}.

\item
Recall that the proof is a counterexample.  This allows us to choose any
value for $F$ (that it could possibly take) in the structures \bf{A} and
\bf{B}.

\item
Also note that this proof does \it{not} deal with the sub-property ``$|F|
\ge K$" alone --- it deals with the whole \it{Maximum Matching} property
(Property \ref{def:maxMatchProp}).

\item
In place of \it{Maximum Matching}, we could have used any decision
problem in \bf{P} derived from a maximization problem.
\end{itemize}

\section{Proof of Theorem \ref{thm:esoHierarchy} based on approximability theory}

This proof assumes that $\mb{P} \neq \mb{NP}$.
~We use approximability theory to show that some decision problems in the
class \bf{P}, derived from optimization problems, cannot be
expressed in ESO universal Horn logic.
In particular, we rely on proven upper and lower bounds on the
approximation ratios (defined below) for the Vertex Cover problem.
We disprove the following proposition, which appeared as a theorem in \cite{gradel91}
(also as Theorem 3.2.17 in \cite{gradelAnd7others}):
\begin{proposition}
Assuming a successor relation as a part of the input structure, 
every problem in the class \bf{P} can be expressed in ESO universal Horn
logic.
$\hfill \Box$
\label{thm:esoHornYes}
\end{proposition}
\begin{definition}
\cite{pancoRanjan}
A \bf{NP-optimization} problem $Q$ is a tuple
$Q = \{I_\Q, F_\Q, f_\Q, opt_\Q\}$, where
\begin{enumerate}
\item[(i)] $I_\Q$ is a set of instances to $\Q$,

\item [(ii)]
$F_\Q (I)$ is the set of feasible solutions to instance $I$,

\item[(iii)]
$f_\Q (I, S)$ is the {\em objective function} value to a solution $S \in
F_\Q(I)$ of an instance $I \in I_\Q$.
~It is a function 
$f:\bigcup_{I \in I_\Q}  [\{I\} \times F_\Q(I)] \rightarrow \bb{R}^+_0$ 
(non-negative
reals)\footnote{Of course, when it comes to computer representation,
rational numbers will be used.}, computable in time polynomial in the
size $|I|$ of the domain of $I$,
% \footnote{Strictly speaking, we should use
% $|I|$ here, where $|I|$ is the length of the representation of $I$.
% ~However, $|I|$ is polynomial in $|A|$, hence we can use  $|A|$.},

\item[(iv)]
For an instance $I \in I_\Q$, $opt_\Q (I)$ is either the minimum or
maximum possible value that can be obtained for the objective function,
taken over all feasible solutions in $F_\Q(I)$.

$\ds
opt_\Q (I) = \max_{S \in F_\Q(I)} f_\Q (I, S)$ (for NP-maximization
problems),

$\ds
opt_\Q (I) = \min_{S \in F_\Q(I)} f_\Q (I, S)$ (for NP-minimization
problems),

\item[(v)]
The following decision problem is in \bf{NP}:
\it{Given an instance $I$ and a non-negative integer $K$, is there a feasible
solution $S \in F_\Q (I)$, such that $f_\Q (I, S) \ge K$ (for an
NP-maximization problem), or $f_\Q (I, S) \le K$ (for an
NP-minimization problem)?}
\end{enumerate}

The set of all such NP-optimization problems is the $\mb{NP_{opt}}$ class.
$\hfill \Box$
\label{def:npOptProblem}
\end{definition}
Note: \it{Feasibility} as defined in point (ii) above has nothing to do
with any upper or lower bounds on the objective function value $f_\Q$.
~It only concerns the BFC, see Sec. \ref{sec:bfcOfc}. ~This also applies
to Def. \ref{def:decisionProblem} below.

\begin{definition}
\bf{Decision versions}.
% \newline
See Def. \ref{def:npOptProblem}.
Given a non-negative integer $K$ and an instance $I \in I_\Q$,
the decision version of an NP-optimization problem $Q$ asks whether there is a
feasible solution $S \in F_\Q (I)$, such that 
$f_\Q (I, S) \ge K$ (if $Q$ is a maximization problem), or 
$f_\Q (I, S) \le K$ (if $Q$ is a minimization problem).
$\hfill \Box$
\label{def:decisionProblem}
\end{definition}

\begin{theorem}
(See 
% \cite{gradel91}, or
Theorem 3.2.17 in Page 147 of \cite{gradelAnd7others}.)
If a decision problem $P_1$ can be represented by an ESO universal Horn
sentence, then 
$P_1$ can solved in polynomial time.
$\hfill \Box$
\label{thm:gradel}
\end{theorem}
The above theorem follows from the fact that HornSAT (satisfiability of a
propositional Horn formula) can be solved in \textsc{ptime}\footnote{In this
paper, we only use \textsc{ptime} as a shorthand for polynomial time, not to
represent the problems in the class \bf{P}.}
\cite{jonesLaaser76}.

% The goal is to determine values for the second order quantifiers (the
% unknowns in the decision problem) in \textsc{ptime}; this would imply
% that the problem represented by the existential second order (ESO)
% sentence is solvable in \textsc{ptime}.

\begin{definition}
\bf{Approximation ratio}.
Given an instance $I$ of an optimization problem $Q$ (as in Def.
\ref{def:npOptProblem}) and a feasible solution $S \in F_\Q (I)$, 
the {approximation ratio} is defined as 
$\ds \frac {f_\Q (I, S)} {f^{(opt)}_\Q (I)}$ for minimization problems
and 
$\ds \frac {f^{(opt)}_\Q (I)} {f_\Q (I, S)}$ for maximization problems.
We assume that $f^{(opt)}_\Q (I) > 0$, where
$f^{(opt)}_\Q (I)$ is the optimal solution value to instance $I$.
$\hfill \Box$
\label{def:approxRatio}
\end{definition}

\begin{definition}
\bf{Approximation problems $\mb{A^T_\Q}~(K_\Q)$ and $\mb{A^B_\Q}~(L_\Q)$}.

\ul{Given}:
An instance $I$ of a problem $Q \in \mb{NP_{opt}}$ (cf. Def. \ref{def:npOptProblem}).

(Upper bound) Problem $\mb{A^T_\Q} ~ (K_\Q)$.
Is there a feasible solution $S \in F_\Q
(I)$, such that the approximation ratio is at most $K_\Q$?

(Lower bound) Problem $\mb{A^B_\Q} ~ (L_\Q)$.
Is there a feasible solution $S \in F_\Q (I)$, 
such that the approximation ratio\footnote{Yes indeed, both problems in 
Def.  \ref{def:approxProblem} have the phrase ``at most" in their
definitions.  This is \bf{not} a typographical error.}
is at most $L_\Q$?

($1 \le L_\Q < K_\Q$.) 
\label{def:approxProblem}
$\hfill \Box$
\end{definition}

Some remarks about Def. \ref{def:approxProblem} (also see the Appendix): 
\begin{remark}
A \it{feasible solution} as mentioned in the definition of $A^T_\Q
~(K_\Q)$ and $A^B_\Q ~(L_\Q)$ has nothing to do with the
objective function value $f_\Q$.
~Hence it has nothing to do with the approximation ratio either.  This
feasibility only concerns the constraints in the BFC. ~See Sec.
\ref{sec:bfcOfc}.

Def. \ref{def:approxProblem} is irrelevant for problems in
\bf{P}, hence we assume that the decision version of Problem $Q$ 
is not known to be polynomially solvable.
$\hfill \Box$
\end{remark}

\begin{definition}
\bf{The APX class}.
~Consider a problem $Q$ from the class
$\mb{NP_{opt}}$ (cf. Def. \ref{def:npOptProblem}).  Then $Q$ 
is defined to be a member of APX iff the decision problem $A^T_\Q ~(K_\Q)$ can be
solved in \textsc{ptime} for some constant $K_\Q \ge 1$.
$\hfill \Box$
\label{def:APX}
\end{definition}
% Of course, if $A^T_\Q$ can be solved in polynomial time and $K_\Q = 1$,
% then $Q \in \mb{P_{opt}}$.

\begin{remark}
Strictly speaking, the upper bound problem $A^T_\Q ~(K_\Q)$ depends on
the parameter $K_\Q$. ~However, for ease of discussion, let us assume
that $K_\Q$ is pre-defined for $A^T_\Q$; assume that $K_\Q$ is
``in-built" into $A^T_\Q$. ~Thus henceforth, we simply write $A^T_\Q$
rather than $A^T_\Q ~(K_\Q)$. 

~We make a similar assumption for the lower bound problem $A^B_\Q$ and the
related parameter $L_\Q$; we will write $A^B_\Q$ rather than $A^B_\Q
~(L_\Q)$. 

For a proven upper bound $K_\Q$, $A^T_\Q~(K_\Q)$ is \textsc{ptime}
solvable.  Similarly, for a proven lower bound $L_\Q$,
$A^B_\Q~(L_\Q)$ is NP-complete.
$\hfill \Box$
\end{remark}

Example: Steiner Tree is a problem in APX, with parameter $K_\Q = 2$.
% The optimization version of this problem can be described as follows:
% 
% \it{
% \ul{Given}: A graph
% $G = (V, E)$, a subset $S$ of $V$ and edge weights $w: E \to \bb{N}$.
% \newline
% \ul{To Do}:
% Find a subgraph $T$ of $G$ such that 
% \newline
% (i) The BFC: ~$T$ includes all the vertices in $S$, $T$ is a tree, and 
% \newline
% (ii) The objective function: ~The sum of the weights of the edges in $T$
% is to be minimized.
% }
% 
Hence given an instance $I$ of the Steiner tree problem, there is a
\textsc{ptime} algorithm which can decide whether there is a feasible
solution to $I$ with an approximation ratio of at most two.

However, $K_\Q$ is an \it{upper bound} on the approximation ratio; we
know that there is a \textsc{ptime} algorithm that guarantees a feasible
solution within this upper bound.
Furthermore, many problems (such as Vertex Cover, below) in the class APX
also have a proven \it{lower bound} $L_\Q$.
~For such problems, the approximation problem ${A^B_\Q}$ is known to
be \bf{NP}-complete.

\begin{problem}
\bf{Vertex Cover (optimization version)}.

\ul{Given}: A graph $G = (V,E)$.

\ul{To do}:  Mark vertices in $V$ in such a way that
\newline
(i)  ~For any $x, y \in V$ where $x < y$, if 
$(x, y) \in E$, then either $x$ of $y$ is marked, and 
\newline
(ii)  ~The number of marked vertices is to be minimized.

\ul{Feasible solution to the problem}:
\newline
A set of marked vertices that obeys property (i).
\label{prob:vertexCover}
$\hfill \Box$
\end{problem}

Let $A^T_{vc}$ and $A^B_{vc}$ be the approximation problems
$\mb{A^T_\Q}$ and $\mb{A^B_\Q}$ respectively (from Def.
\ref{def:approxProblem}), applied to Vertex Cover. 

$A^T_{vc}$ and $A^B_{vc}$ are members of the class \bf{NP}; in particular,
$A^T_{vc}$ is a member of the class \bf{P}, and $A^B_{vc}$ is \bf{NP}-complete.

The upper bound $K_\Q = K_{vc} = 2$ and the lower bound
$L_\Q = L_{vc} = 1.36006$ \cite{dinurSafra2005}.
~In fact, for $L_{vc}$, one could use any number $\gamma$ such that
$1 < \gamma \le 1.36006$.
~For ease of representation in an ESO formula, let us use $\gamma = 1.3$.

Here is the \label{page:coreArg} core argument:
\begin{lemma}
Proposition \ref{thm:esoHornYes} is incorrect.
\label{lem:esoHornNo}
$\hfill \Box$
\end{lemma}

\begin{proof}
Since $A^T_{vc} \in \mb{P}$, Theorem \ref{thm:esoHornYes} can be
applied.
Then it follows that $A^T_{vc}$ can be expressed as an ESO universal Horn
sentence $\phi_1$.  This sentence can be written as a conjunction of two
sub-formulae:
\begin{enumerate}
\item
(the feasibility condition)
One stating that every edge is covered, which is easily written as
\[
\exists S ~ \forall x \forall y ~ \neg E(x,y) \vee S(x) \vee S(y)
\]
 and
\item
 (the approximation ratio condition)
Another stating that the approximation ratio $\alpha (I, S)$ is at most
two, that is,
\[
\alpha (I,S) \le 2.
\]
\end{enumerate}

However, if we replace the sub-formula for the condition $\alpha (I, S)
\le 2$ with the sub-formula for $\alpha (I, S) \le 1.3$, we obtain a new ESO
universal Horn sentence $\phi_2$ which captures $A^B_{vc}$. 

Though $\ds A^B_{vc}$ is known to be \bf{NP}-complete, we have just
obtained an ESO universal Horn sentence ($\ds \phi_2$) for this problem,
implying that $\ds A^B_{vc}$ can be decided in \textsc{ptime}, when we
apply Theorem \ref{thm:gradel}.
But this contradicts our assumption that \bf{P} $\neq$ \bf{NP}.
$\hfill \Box$
\end{proof}

Hence we have shown the following:
\begin{lemma}
Assuming that \bf{P} $\neq$ \bf{NP}, and assuming a successor relation as
a part of the input vocabulary,
$A^T_{vc}$, a problem solvable in \textsc{ptime}, cannot be expressed as an
ESO universal Horn sentence.
$\hfill \Box$
\end{lemma}
This concludes the second proof of Theorem \ref{thm:esoHierarchy}.

% \section{Proof sketch for Theorem \ref{thm:faginFalse}}
% 
% Problems in \bf{NP} (in particular, those derived from maximization problems)
% that have a constraint ``$|F| \ge K$" on the solution
% value, cannot be expressed with a universal first order part, since this
% violates the ``preservation under substructure" property of universal
% formulae.
% For a given structure $A$ and lower bound $K$, suppose it satisfies an ESO
% universal sentence $\phi$.
% If we consider a substructure $B$ of $A$ that is ``small enough", its
% solution value will violate the lower bound $K$ and hence $\phi$.
% 
% It is known that linear ordering of an input structure can be defined by
% ESO universal formulae (Chapter 7 in \cite{immerman}).
% Hence Theorem \ref{thm:faginFalse} is also true where the input structure
% is unordered.

\section{Separation between machine level and structure level}\label{sec:future}

We note a discrepancy between the results in Theorems
\ref{thm:esoHierarchy} and \ref{thm:esoHornYes}.  We observe that
there is a fundamental difference between the types of expressions dealt
with by the two theorems.
Theorem \ref{thm:esoHierarchy} is about \it{structure level expressions}
(SLE) which deal ``directly" with properties of mathematical structures,
whereas Theorem \ref{thm:esoHornYes} is about \it{machine level
expressions} (MLE) which represent encodings of machine level
computations of properties; these merely encode the \bf{steps}
(\bf{moves}) of the computation.

We note that (ESO universal Horn) = \bf{P} at the \it{machine level},
whereas (ESO universal Horn) $\subset$ \bf{P} at the \it{structure
level}.

Observe that in the two proofs, we \bf{do not} use a Turing machine; the
expressions \bf{are not} the result of an output from a Turing machine computation.

\subsection{The importance of structure level over machine
level}\label{sec:importance}

The precise definitions for \it{structure level} and \it{machine level} are
provided later in the paper --- but first, we describe intuitive
notions.

A few observations are in order:
\begin{enumerate}
\item
Descriptive Complexity is a very useful tool at the structure level,
\it{not} at the machine level.
\item
A large subset of problems in the class \bf{NP} is derived from
optimization problems, hence this is a important sub-class.
(Problem \ref{def:matching} has also been derived from an
optimization problem.)
\item
In spite of the important sub-class described above, we were unable to
locate any \it{structure level} expression for such problems in the
literature (although one can construct \it{machine level} expressions
using the methods in \cite{fagin}, Chapter 7 of \cite{immerman} or
\cite{immerman86}).  This is a key motivation for our line of research.
\end{enumerate}

From a mathematician's viewpoint, SLE's are the natural choice to
express a problem/property, \it{not} the MLE's.  To logically express a
property of a given mathematical structure, what is the need to put it
through computation first? 

To write an MLE for a given property $\mc{P}$, we first need to design a
TM which will decide whether $\mc{P}$ is true (for any given input).  But
this means that the space/time complexity to decide $\mc{P}$ is already
known, hence it is pointless writing an MLE for $\mc{P}$.

To elaborate on the item numbered one above:
The usefulness of Descriptive Complexity lies in the fact that if one
could write a structure level expression in a certain logic for a
decision problem $\mc{D}$, then it can be deduced that $\mc{D}$ requires
a certain amount of resource (time and/or space) for its computation;
that is, we can recognize the complexity class that $\mc{D}$ belongs to.

However, our main result (Theorem \ref{thm:esoHierarchy})
shows that even though a large class of problems
(those derived from optimization) is \textsc{ptime} solvable, they cannot
be expressed in ESO universal Horn logic at the structure level.  Hence
for such problems\footnote{Of course, we fully acknowledge contributions
in the literature that provided machine level expressions --- but that is
only one side of the story.}, the usefulness of Descriptive Complexity is
lost while using this particular logic.

\subsection{An informal definition of the two levels}

\bf{Machine level expressions} (MLE).
These express the \it{computation} of a property (in a Turing machine,
for instance).
Here, \it{computation of a property} is defined as \it{checking whether a
property is true using a computation device}, for a given input.

Such an expression encodes the computation steps of a Turing machine
(TM), for example, as in the proof by Fagin \cite{fagin} that ESO logic
captures the class \bf{NP} and the proof by Immerman \cite{immerman86} that LFP
(under ordered structures) captures the class \bf{P}.
~These expressions state that 
\newline
(i) if we go through the steps of a computation
device such as a TM, and 
\newline
(ii) if the TM finally reaches an accepting state,
\newline
then the property must be true for the given input structure.

For example, to write a machine level expression for \it{Hamiltonian
Cycle} (HC), we should first design a TM that solves HC, then encode the
steps of the TM in an expression.

\bf{Structure level expressions} (SLE).
On the other hand, {structure level} logical expressions do not
express the \it{computation} of a property; rather, they directly express
the truth of a property of a mathematical structure.  \it{No
computation is involved in such a description}.

\it{Machine level} expressions for the class \bf{NP} (class \bf{P}) were provided
by Fagin \cite{fagin} (Gr\"{a}del \cite{gradel91}) respectively.

\subsection{Implications}

Since the structure level is indeed proven to be different from the machine
level, we can then state the following:
\begin{remark}
At the structure level, the validity of Fagin's theorem is completely
open --- it is unknown whether ESO sentences can characterize the class
\bf{NP} at this level.
$\hfill \Box$
\end{remark}

In the last 40 years, important contributions have been made at the
``machine level", such as Fagin's theorem and the Immerman-Vardi theorem.
However, these do not represent the full picture.

The formal mathematical definitions for
concepts such as MLE and SLE explain accurately as to why 
Theorems \ref{thm:esoHierarchy} and \ref{thm:esoHornYes} are correct in
their respective ``domains" (the machine computation domain and the
mathematical structure domain).

We now provide the formal definitions of MLE and SLE.

\subsection{Definitions of machine level and structure level}

Given a structure $\mc{A}$ as input, defined on a first order vocabulary
$\sigma$, we provide the following definitions:

\begin{definition}
\bf{Higher order predicates} are those that are not defined in the first
order vocabulary of the input structure.
\label{def:HOpred}
\end{definition}
\begin{definition}
% \begin{enumerate}
% \item[(a)] 
(a) Let $\mc{E^L (P)}$  be the set of expressions in logic $\mc{L}$ such
that a structure $\mc{A}$ possesses property $\mc{P}$ iff 
$\mc{A} \models \mc{E}$ for every expression $\mc{E} \in \mc{E^L (P)}$. 

% \item[(b)] 
(b) We let $\mc{E}_m^{\mc{L}} (\mc{P})$ be the set of \bf{machine level
expressions} (MLE) for $\mc{P}$.
~Naturally, $\mc{E}_m^{\mc{L}} (\mc{P}) \subseteq \mc{E^L (P)}$.
~We define $\mc{E}_m^{\mc{L}} (\mc{P})$ as follows.

Let $T_i (\mc{P})$ be
Turing machines that have at least one halting computation iff
$\mc{A}$ possesses property $\mc{P}$ where $i \ge 1$.
~Hence when $\mc{A}$ satisfies $\mc{P}$, let $C_{i,j}$ be the $j^{th}$
halting computation of $T_i (\mc{P})$ where $j \ge 1$.
Over the set of all Turing machines \{$T_i (\mc{P})~|~i \ge 1\}$, the set
of halting computations is the set $\{C_{i,j}~|~i,j \ge 1\}$.
~Let $\mc{E} (C_{i,j})$ be an encoding of $C_{i,j}$ in $\mc{L}$.

Then $\mc{E}_m^{\mc{L}} (\mc{P})$  = $\{\mc{E} (C_{i,j})~|~i,j \ge 1\}$.
\newline
[We know from (a) that $\mc{A} \models \mc{E}$ for every $\mc{E} \in 
\mc{E}_m^{\mc{L}} (\mc{P})$.]

The predicates that encode machine computations are of higher order (cf.
Def. \ref{def:HOpred}).
% \end{enumerate}
\label{def:MLE}
$\hfill \Box$
\end{definition}
Above, we say ``at least one halting computation" to take into account
the fact that $T_i$ could be a non-deterministic TM.

Next, we define structure level expressions, as a subset of $\mc{E^L (P)}$.
\begin{definition}
Consider the set $\mc{E^L (P)}$ in Def. \ref{def:MLE}.
We let $\mc{E}_s^{\mc{L}} (\mc{P})$ be the set of \bf{structure level
expressions} (SLE) for $\mc{P}$.
~Naturally, $\mc{E}_s^{\mc{L}} (\mc{P}) \subseteq \mc{E^L (P)}$.
~We define $\mc{E}_s^{\mc{L}} (\mc{P})$ as follows:
For every expression $\mc{E} \in \mc{E}_s^{\mc{L}} (\mc{P})$,
\newline
(i) the higher order predicates of $\mc{E}$ range over tuples from the
universe of $\mc{A}$,  and 
\newline
(ii) the variables and first order predicate symbols in $\mc{E}$ are from $\sigma$.
\label{def:SLE}
$\hfill \Box$
\end{definition}

\subsection{Distinguishing MLE from an SLE}

Although Def. \ref{def:MLE} explains how to generate an MLE, it doesn't
tell us how to recognise an MLE when we see one; it doesn't tell us how
to distinguish an MLE from an SLE.
~This can be achieved quite easily since the vocabulary of the two
expression types are different.  For an SLE, the language consists of 
variables and relation symbols from the input vocabulary $\sigma$ of
structure $\mc{A}$ (as in Def. \ref{def:SLE}).

However for an MLE, the language consists of \bf{machine encodings} for
\newline
(i) the tape symbols of a TM;
\newline
(ii) the tape head movement and transitions of the TM; and
\newline
(iii) the variables and relation symbols in $\sigma$ as well as the
higher order predicate symbols used to encode the computation.

\section{Conclusion: Easier done than said}

For the class \bf{P}, it is a case of ``Easier done than said".  One can
let the computation of a decision problem run through a Turing machine
and obtain an expression in ESO universal Horn logic.  However,
\it{without} computation, ESO universal Horn is insufficient to capture
\bf{P} (we need a stronger logic), as the two proofs of Theorem
\ref{thm:esoHierarchy} show.

\section{Acknowledgements}
Nerio Borges, Anuj Dawar, Ian Hodkinson, Leonid Libkin and Anand Pillay.

This paper was presented at the Asian Logic Conference, Wellington, NZ,
in December 2011.

\appendix

\section{More on Definition \ref{def:approxProblem}}

Some readers who read a draft of this paper, after reading
Def. \ref{def:approxProblem}, remarked, ``These problems are very easy".

Well, I doubt if these problems can be categorised as ``very easy".
However, points to note include the following:
\begin{enumerate}
\item
A \it{feasible solution} means that the solution obeys the BFC (cf. Sec.
\ref{sec:bfcOfc}) --- feasibility here has nothing to do with the
objective function nor the approximation ratio;

\item
If $\mb{P} = \mb{NP}$, then $\ds K_\Q = 1$ and $\ds L_\Q$ has no meaning.

Conversely, if we can design an algorithm that solves $\ds A^T_{\Q}~ (K_\Q)$ in
polynomial time for $\ds K_\Q = 1$, then $\ds \mb{P} = \mb{NP}$;

\item
If the answer is YES (to either of the problems defined), then a
``feasible" solution $S$ that obeys the bound ($K_\Q$ or $L_\Q$) must be
returned;

\item
The approximation ratio is \it{not} given to us; we are only given a
\it{bound} on the approximation ratio ($K_\Q$ or $L_\Q$, as the case may be); 

\item
Hence, for the upper bound problem $A^T_\Q~(K_\Q)$, to answer the
question with a YES or NO within polynomial time, over all instances, we
should develop a polynomial time algorithm for $Q$ (which is not a
trivial task for many a problem $Q \in \mb{NP_{opt}}$);
and

\item
For a proven upper bound $K_\Q$, $A^T_\Q~(K_\Q)$ is \textsc{ptime}
solvable.  Similarly, for a proven lower bound $L_\Q$,
$A^B_\Q~(L_\Q)$ is NP-complete.
~Note that $1 \le L_\Q < K_\Q$. ~It is the gap between $L_\Q$ and $K_\Q$
that we seek to exploit.
$\hfill \Box$

\end{enumerate}

\section{More about the Core Argument on Page \pageref{page:coreArg}}

A reader raised a concern as to why should the expression for $A^T_{vc}$
consist of a sub-formula for the condition $\alpha (I,S) \le 2$.

The answer is, it is \it{not} a question of \bf{should}, but a question
of \bf{can}. 
~ Yes, we \it{can} express $A^T_{vc}$ as a conjunction of the
(sub-formulae) for the two conditions --- the feasibility condition, and
the approximation ratio condition. 

Once this is done, we can modify this to an expression for $A^B_{vc}$.

Note that the expressions in the proof of Lemma \ref{lem:esoHornNo} deal
purely with the mathematical structure alone.  No machine computation is
involved in developing the expressions.

\end{document}